\documentclass[envcountsame,envcountsect,runningheads,a4paper]{llncs}

\usepackage{latexsym}
\usepackage{amsmath,amssymb}
\usepackage{wrapfig}
\usepackage{atbeginend}
\usepackage{enumitem}
\usepackage{tikz}
\usepackage{hyperref}

\usetikzlibrary{arrows,automata,backgrounds,calc,chains,fadings,fit,patterns,positioning,shapes.geometric,shapes.symbols,through,trees,decorations.markings,decorations.text}

\AfterBegin{definition}{\normalfont}

\newcommand{\tup}[1]{\langle#1\rangle}
\newcommand{\pair}[2]{\tup{#1,#2}}

\newcommand{\msf}{\mathsf}
\newcommand{\mbb}{\mathbb}
\newcommand{\mbf}{\mathbf}
\newcommand{\mbs}{\boldsymbol}
\newcommand{\mcl}{\mathcal}
\newcommand{\mrm}{\mathrm}

\newcommand{\nat}{\mbb{N}}
\newcommand{\rat}{\mbb{Q}}
\newcommand{\zz}{\mbb{Z}}

\newcommand{\myex}[2]{\exists #1.\,#2}

\newcommand{\emptyword}{\varepsilon}

\newcommand{\lt}{<}
\newcommand{\gt}{>}

\newcommand{\length}[1]{|#1|}

\newcommand{\fap}[2]{#1(#2)}
\newcommand{\fapn}[2]{\fap{#1^{#2}}}
\newcommand{\bfap}[3]{\fap{#1}{#2,#3}}

\newcommand{\where}{\mid}

\newcommand{\two}{\mbf{2}}
\newcommand{\str}[1]{#1^{\nat}}

\newcommand{\fstred}{\ge}
\newcommand{\fstredstrict}{>}
\newcommand{\fstconv}{\equiv}
\newcommand{\convclass}[1]{[#1]}

\newcommand{\botdeg}{\mbs{0}}

\newcommand{\seq}[1]{\langle#1\rangle}
\newcommand{\sshift}{\mcl{S}}
\newcommand{\shift}{\fapn{\sshift}}

\newcommand{\zloops}{\fap{\msf{Z}}}
\newcommand{\pto}{\rightharpoonup}

\newcommand{\wof}[2]{#1 \cdot #2}

\newcommand{\wprod}[2]{#1 \otimes #2}

\newcommand{\sdpf}{\varPhi}
\newcommand{\dpf}{\fap{\sdpf}}

\newcommand{\floor}[1]{\left\lfloor{#1}\right\rfloor}

\newcommand{\zip}{\msf{zip}}

\newcommand{\lfst}{$\text{FST}^{\star}$}

\newcommand{\scyc}{\varphi}
\newcommand{\scycp}{\varphi'}
\newcommand{\cyc}{\bfap{\scyc}}
\newcommand{\cycp}{\bfap{\scycp}}

\newcommand{\prefixof}{\sqsubseteq}

\newcommand{\sidifnn}{\mrm{id_{\ge 0}}}
\newcommand{\idifnn}{\fap{\sidifnn}}

\newcommand{\morse}{\msf{T}}
\newcommand{\perioddoubling}{\msf{P}}

\tikzset{state/.style={draw=black,ellipse,inner sep=.6mm,outer sep=.5mm}}
\tikzset{default/.style={->,>=stealth',shorten >=1pt,shorten <= 1pt,auto,node di
stance=2cm,semithick}}

\tikzstyle{roundNode}=[gyellow,thick,circle,minimum size=4mm,inner sep=0.5mm]

\tikzset{state/.style={draw=black,ellipse,inner sep=.6mm,outer sep=.5mm}}
\tikzset{default/.style={->,>=stealth',shorten >=1pt,shorten <= 1pt,auto,node di
stance=2cm,semithick}}
\tikzset{bl/.style={below left of=#1,yshift=3mm}}
\tikzset{br/.style={below right of=#1,yshift=3mm}}

\tikzset{lbr/.style={below right,inner sep=0.5mm}}
\tikzset{lar/.style={above right,inner sep=0.5mm}}
\tikzset{lbl/.style={below left,inner sep=0.5mm}}
\tikzset{lal/.style={above left,inner sep=0.5mm}}

\tikzset{tloop/.style={out=60,in=120,looseness=5}}
\tikzset{bloop/.style={out=-60,in=-120,looseness=5}}
\tikzset{lloop/.style={out=210,in=150,looseness=5}}
\tikzset{rloop/.style={out=-30,in=30,looseness=5}}

\tikzset{lhead/.style={at=(#1.west),anchor=east,xshift=0cm,inner sep=.5mm}}
\tikzset{rhead/.style={at=(#1.east),anchor=west,xshift=0cm,inner sep=.5mm}}
\tikzset{bhead/.style={at=(#1.south),anchor=north,xshift=0cm,inner sep=.5mm}}

\tikzstyle{gyellow}=[draw=brown!80,top color=yellow!50,bottom color=yellow!20!brown]
\tikzstyle{gblue}=[draw=blue!50,top color=white,bottom color=blue!60]
\tikzstyle{gred}=[draw=red!50,top color=white,bottom color=red!60]
\tikzstyle{ggreen}=[draw=green!70!black,top color=white,bottom color=green!80!black]

\tikzstyle{roundNode}=[gyellow,thick,circle,minimum size=4mm,inner sep=0.5mm]

\begin{document}

\title{The Degree of Squares is an Atom \\ 
  (Extended Version\thanks{This is an extended version of \cite{endr:grab:hend:zant:2015}.})}
\author{
  J\"{o}rg Endrullis \inst{1}
  \and
  Clemens Grabmayer \inst{1}
  \and
  Dimitri Hendriks \inst{1}
  \and
  Hans~Zantema \inst{2,3}
}

\institute{
  Department of Computer Science, VU University Amsterdam 
  \and
  Department of Computer Science, Eindhoven University of Technology 
  \and
  Institute for Computing and Information Science, Radboud University Nijmegen 
}

\date{}
\maketitle

\begin{abstract}
  We answer an open question in the theory of degrees of infinite sequences with respect to transducibility
  by finite-state transducers.
  An initial study of this partial order of degrees
  was carried out in~\cite{endr:hend:klop:2011},
  but many basic questions remain unanswered.
  One of the central questions concerns the existence of atom degrees,
  other than the degree of the `identity sequence'
  $1 0^0 1 0^1 1 0^2 1 0^3 \cdots$.
  A degree is called an `atom' if below it 
  there is only the bottom degree~$\botdeg$, which consists of the ultimately periodic sequences.
  We show that also the degree of the `squares sequence' $1 0^0 1 0^1 1 0^4 1 0^9 1 0^{16}\cdots$ is an atom.

  As the main tool for this result 
  we characterise 
  the transducts of `spiralling' sequences and their degrees.
  We use this to show that   
  every transduct of a `polynomial sequence' either is in~$\botdeg$
  or can be transduced back to a polynomial sequence 
  for a polynomial of the same order. 
\end{abstract}

\section{Introduction}\label{sec:intro}
Finite-state 
transducers are ubiquitous in computer science,
but little is known about the transducibility relation they induce on infinite sequences.
A finite-state transducer (FST) is a deterministic finite automaton 
which reads the input sequence letter by letter,
in each step producing an output word and changing its state.
An example of an FST is depicted in Figure~\ref{fig:fst:diff},
where we write `$a | w$' along the transitions to indicate that the input
letter is $a$ and the output word is $w$.
For example, it 
transduces the Thue-Morse sequence 
$\morse = 0110 1001 1001 0110 \cdots$ 
to the period doubling sequence $\perioddoubling = 101 1 101 0 101 1 101 1 \cdots$.

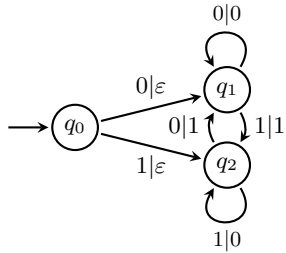
\begin{figure}
  \centering
  \begin{tikzpicture}[thick,>=stealth,node distance=20mm,
    node/.style={circle,fill=none,draw=black,inner sep=.4ex,minimum size=6mm,outer sep=.5mm}, l/.style={scale=.9}]
    \node (start) {};
    \node (q0) [node,right of=start,node distance=10mm] {$q_{0}$};
    \node (q1) [node,right of=q0,yshift=5mm] {$q_{1}$};
    \node (q2) [node,right of=q0,yshift=-5mm] {$q_{2}$};
    \draw [->] (start) -- (q0);
    \draw [->] (q0) -- (q1) node [midway, above] {$0|\emptyword$};
    \draw [->] (q0) -- (q2) node [midway, below] {$1|\emptyword$};
    \draw [->] (q1) to[bend left=30] node [midway, right] {$1 | 1$} (q2) ;
    \draw [->] (q2) to[bend left=30] node [midway, left] {$0 | 1$} (q1);
    \draw [->] (q1) to[out=60,in=120,looseness=5] node [l,above] {$0|0$} (q1);
    \draw [->] (q2) to[out=-60,in=-120,looseness=5] node [l,below] {$1|0$} (q2);
  \end{tikzpicture}
  \caption{A finite-state transducer realising the sum of consecutive bits modulo~$2$.}
  \label{fig:fst:diff}
\end{figure}

We are interested in transductions of infinite sequences. 
We say that a sequence $\sigma$ is \emph{transducible} to a sequence $\tau$, 
$\sigma \fstred \tau$, if there exists an FST that transforms $\sigma$ into $\tau$.
The relation $\fstred$ is a preorder on the set $\str{\Sigma}$ of infinite sequences, 
which
induces an equivalence relation~$\fstconv$ on~$\str{\Sigma}$,
and a partial order on the set of \emph{(FST) degrees}, that is, 
the equivalence classes with respect to~$\fstconv$.

So we have $\morse \fstred \perioddoubling$. 
Also the back transformation can be realised by an FST,
$\perioddoubling \fstred \morse$. 
Hence the sequences are equivalent, $\morse \fstconv \perioddoubling$,
and are in the same degree.

The bottom degree $\botdeg$ is formed by the ultimately periodic sequences,
that is, all sequences of the form~$u v v v \cdots$ for finite words~$u,v$ with $v$ non-empty.
Every infinite sequence can be transduced to any ultimately periodic sequence.

There is a clear analogy between degrees induced by transducibility 
and the recursion-theoretic degrees of unsolvability (Turing degrees).
Hence many of the problems settled for Turing degrees, 
predominantly in the 1940s, 50s and 60s, can be asked again for 
FST~degrees.

\begin{figure}
  \begin{center}
  \begin{minipage}{.2\textwidth}
  \begin{center}
    \begin{tikzpicture}[thick,>=stealth,
      node/.style={circle,fill=black,draw=none,ggreen,state},scale=.7]
    
      \node (evp) [node,inner sep=0,minimum size=4.75mm] at (0mm,0mm) {$\botdeg$};
      \node (a) [node,minimum size=4mm] at (-12mm,12mm) {};
      \node (b) [node,minimum size=4mm] at (12mm,12mm) {};
      \node (c) [node,minimum size=4mm] at (0mm,24mm) {};
    
      \begin{scope}[very thick,->]
      \draw (c) -- (a); \draw (a) -- (evp);
      \draw (c) -- (b); \draw (b) -- (evp);
      \end{scope}
    \end{tikzpicture}
  \end{center}
  \end{minipage}~%
  \begin{minipage}{.2\textwidth}
  \begin{center}
    \begin{tikzpicture}[thick,>=stealth,
      node/.style={circle,fill=black,draw=none,ggreen,state},scale=.7]
    
      \node (evp) [node,minimum size=4.57mm,inner sep=0] at (0mm,0mm) {$\botdeg$};
      \node (a) [node,minimum size=4mm] at (0mm,12mm) {};
      \node (b) [node,minimum size=4mm] at (0mm,24mm) {};
    
      \begin{scope}[very thick,->]
      \draw (b) -- (a); \draw (a) -- (evp);
      \end{scope}
    \end{tikzpicture}
  \end{center}
  \end{minipage}
  \caption{%
      Possible structures in the hierarchy 
      (no intermediate points on the arrows).
  }
  \end{center}
  \label{fig:diamond:line}
\end{figure}
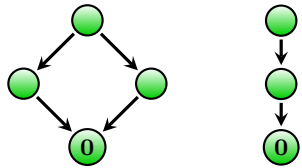

Some initial results on FST degrees have been obtained in~\cite{endr:hend:klop:2011}:
the partial order of degrees is not dense, not well-founded, there exist no maximal degrees,
and a set of degrees has an upper bound if and only if the set is countable.
The morphic degrees, and the computable degrees form 
subhierarchies.
Also shown in~\cite{endr:hend:klop:2011} is the existence of an `atom' degree,
namely the degree of~$1 0^0 1 0^1 1 0^2 1 0^3 \cdots$.%
\setcounter{footnote}{0}%
\footnote{%
  In~\cite{endr:hend:klop:2011} atom degrees were called `prime degrees'.
  We prefer the more general notion of `atom' because prime factorisation does not hold,
  see Theorem~\ref{thm:no:minimal:below}.
}
A degree $D \ne \botdeg$ is an atom if there exists no degree 
strictly between $D$ and the bottom degree $\botdeg$.
Thus every transduct of a sequence in an atom degree $D$ 
is either in $\botdeg$ or still in $D$.
%
The following questions, which have been answered for Turing degrees,
remain open for the FST degrees:
\begin{enumerate}

  \item
    How many atom degrees exist? 

  \item
    When do two degrees have a supremum (or infimum)? 
    In particular, are there pairs of degrees without a supremum (infimum)?

  \item
    Do the configurations displayed in Figure~\ref{fig:diamond:line} exist?

  \item 
    Can every finite distributive lattice  
    be embedded in the hierarchy? 



\end{enumerate}

At the British Colloquium for Theoretical Computer Science 2014, Jeffrey Shallit offered 
\textsterling 20 for each of the following questions, see~\cite{shal:2014}:
\begin{enumerate}[label=(\alph*)]
  \item
    Is the degree of the Thue-Morse sequence an atom?
  \item 
    Are there are atom degrees other than that of
    $1 0^0 1 0^1 1 0^2 1 0^3 \cdots$?
    \label{item:shallit:prime}
\end{enumerate}
We answer \ref{item:shallit:prime} by showing that 
the degree of the `squares' $1 0^0 1 0^1 1 0^4 1 0^9 \cdots$ is an atom.
The main tool that we use in the proof is 
a characterisation of transducts of `spiralling' sequences (Theorem~\ref{thm:transducts}),
which has the following consequence (Proposition~\ref{prop:poly-transducts}):
For all $k \ge 1$ it holds that
every transduct~$\sigma \not\in \botdeg$ of the sequence $\seq{n^k} = 1 0^{0^k} 1 0^{1^k} 1 0^{2^k} 1 0^{3^k} \cdots$ 
has, in its turn, a transduct of the form $\seq{p(n)} = 1 0^{p(0)} 1 0^{p(1)} 1 0^{p(2)} 1 0^{p(3)} \cdots$
where $p(n)$ is a  polynomial also of degree $k$.
This fact not only enables us to show that the degree of the squares sequence is an atom, by using it for $k=2$,
but it also suggests that the analogous result for the degree of $\seq{n^k}$, for arbitrary $k \ge 1$, has come within reach.
For this it would namely suffice to show, for polynomials $p(n)$ of degree~$k$, 
that $\seq{p(n)}$ can be transduced back to $\seq{n^k}$.

We obtain that there is a pair of non-zero degrees whose infimum is~$\botdeg$,
namely the pair of atom degrees of $\seq{n}$ and $\seq{n^2}$.
We moreover use Theorem~\ref{thm:transducts} 
to show that there is a non-atom degree that has no atom degree below it (Theorem~\ref{thm:no:minimal:below}).

\section{Preliminaries}\label{sec:prelims}
We use standard terminology and notation, see, e.g., \cite{saka:03}.
Let $\Sigma$ be an alphabet, 
i.e., a finite non-empty set of symbols.
We denote by $\Sigma^*$ the set of all finite words over~$\Sigma$,
and by $\emptyword$ the empty word. We let $\Sigma^+ = \Sigma^*\setminus\{\emptyword\}$.
The set of infinite sequences over~$\Sigma$ is 
$\str{\Sigma} = \{ \sigma \where \sigma : \nat \to \Sigma \}$
with $\nat = \{0,1,2,\ldots\}$, the set of natural numbers. 
For $u \in \Sigma^*$, we let $u^\omega = u u u \cdots$.
We define $\nat_{<k} = \{0,1,\ldots,k-1\}$.
We let $\Sigma^\infty = \Sigma^* \cup \str{\Sigma}$ denote the set of all words over $\Sigma$.
For $u \in \Sigma^*$ and $v \in \Sigma^\infty$ we write $u \prefixof v$ 
to denote that $u$ is a prefix of $v$, that is, when $v = uv'$ for some $v' \in \Sigma^\infty$.
For $f : \nat \to A$ and $k \in \nat$, the \emph{$k$-th shift} of~$f$
is defined by $\shift{k}{f}(n) = f(n+k)$, for all $n \in \nat$. 

We use the notation $\vec{a}$ to denote a tuple 
$\vec{a} = \tup{a_0,a_1,\ldots,a_{k-1}}$
where the~$a_i$ 
are elements from some set $A$;
the length of $\vec{a}$ is $k$.
Given a tuple~$\vec{a}$ we use $a_i$
for the element indexed~$i$, that is, we start indexing from $0$ onward.
We use the notation $\vec{a}'$ to denote the rotated tuple~%
$\vec{a}' = \tup{a_1,\ldots,a_{k-1}, a_0}$.

\section{Finite-State Transducers and Degrees}\label{sec:degrees}

For a thorough introduction to finite-state transducers, 
we refer to~\cite{saka:03}. 
Here we only consider \emph{complete} \emph{pure} \emph{sequential} finite-state transducers, 
where for every state~$q$ and input letter~$a$ there is precisely one successor state~$\delta(q,a)$,
and the functions realised by these transducers preserve prefixes.
We use $\two = \{0,1\}$ both for the input and the output alphabet.

\begin{definition}\label{def:fst}
  A \emph{finite-state transducer (FST)} is a tuple
  $T = \tup{Q,q_0,\delta,\lambda}$
  where $Q$ is a finite set of states, $q_{0} \in Q$ is the initial state,
  $\delta : Q \times \two \to Q$ is the transition function, 
  and $\lambda : Q \times \two \to \two^*$ is the output function.

  We homomorphically extend the transition function $\delta$ to $Q \times \two^* \to Q$
  and the output function~$\lambda$ to $Q \times \two^\infty \to \two^\infty$ as follows:
  \begin{align*}
    \delta(q,\emptyword) & = q 
    & \delta(q,au) & = \delta(\delta(q,a),u)
    && (q \in Q, a \in \two, u \in \two^*) \\
    \lambda(q,\emptyword) & = \emptyword 
    & \lambda(q,au) & = \lambda(q,a) \cdot \lambda(\delta(q,a),u)
    && (q \in Q, a \in \two, u \in \two^\infty)\,.
  \end{align*}
  The function $T : \two^\infty \to \two^\infty$ \emph{realised} by the FST~$T$
  is defined by $T(u) = \lambda(q_0,u)$,
  for all $u \in \two^\infty$.
\end{definition}

We note that finite state transducers and transduction of infinite sequences
can be implemented as infinitary rewrite systems~\cite{endr:grab:hend:2008,endr:hend:klop:2012,endr:hans:hend:polo:silv:2015}.

\begin{definition}
  Let $T = \tup{Q,q_0,\delta,\lambda}$ be an FST.
  A \emph{zero-loop} in $T$ 
  is a sequence of states $q_{1},\ldots,q_{n}$ with $n > 1$
  such that $q_1 = q_n$ and $q_{i} \ne q_{j}$ 
  for all $i,j$ with $1 \le i < j < n$ 
  and $q_{i+1} = \delta(q_{i},0)$ for all $1 \le i < n$.
  The \emph{length} of the zero-loop is $n-1$.
\end{definition}
Note that there can only be finitely many zero-loops in an FST.
\begin{definition}\label{def:zloops}
  Let $T = \tup{Q,q_0,\delta,\lambda}$ be an FST.
  We define $\zloops{T}$ as the least common multiple of the lengths of all zero-loops of $T$.
\end{definition}

Let $T$ be an FST with states $Q$.
From any state $q \in Q$, after reading the word $0^{|Q|}$, the automaton
must have entered a zero-loop
(by the pigeonhole principle there must be a state repetition).
By definition of $\zloops{T}$, the length $\ell$ of this loop divides $\zloops{T}$;
say $\zloops{T} = d \ell$ for some $d \ge 1$.
As a consequence, reading $0^{|Q|+i\cdot\zloops{T}}$ yields the output $\lambda(q,0^{|Q|})$ followed by 
$d i$ copies of the output of the zero-loop.
This yields a pumping lemma for FSTs, see also~\cite[Lemma~29]{endr:hend:klop:2011}.

\begin{lemma}\label{lem:pumping}
  Let $T = \tup{Q,q_0,\delta,\lambda}$ be an FST.
  For every $q \in Q$ and $n \ge |Q|$
  there exist words $p,c \in \two^*$ such that for all $i \in \nat$,
  $\delta(q,1 0^{n + i \cdot \zloops{T}}) = \delta(q,1 0^n)$ and
  $\lambda(q,1 0^{n + i \cdot \zloops{T}}) = p c^{i}$\,.
\end{lemma}
\begin{proof}
  Let $q\in Q$ and $n \ge |Q|$.
  Then there exist $k,\ell \in \nat$ such that 
  $k + \ell \le n$ and $\delta(q,10^{k}) = \delta(q,10^{k+\ell})$.
  Let $k,\ell$ be minimal with these properties.
  Then $\ell$ is the length of a zero-loop in $T$, 
  and so $\zloops{T} = d \ell$ for some $d \ge 1$.
  We have $\delta(q,10^{k'}) = \delta(q,10^{k'+\ell})$ for all $k' \ge k$,
  and so $\delta(q,10^n) = \delta(q,10^{n + i d \ell}) = \delta(q,10^{n + i \cdot \zloops{T}})$,
  for all $i \in \nat$.
  Let $q' = \delta(q,10^n)$,
  $p = \lambda(q,10^n)$ and $c = \lambda(q',0^{\zloops{T}})$.
  Then we have $\lambda(q,10^{n + i \cdot \zloops{T}}) = \lambda(q,10^n) \cdot \lambda(q',0^{i \cdot \zloops{T}}) = p c^i$,
  for all $i \in \nat$.
  \qed
\end{proof}

\begin{definition}\label{def:fstred}
  Let $T$ be an FST, and let $\sigma,\tau \in \str{\two}$ be infinite sequences.
  We say that $T$ \emph{transduces} $\sigma$ to $\tau$, 
  and that $\tau$ is the \emph{$T$-transduct} of $\sigma$,
  which we denote by $\sigma \fstred_{T} \tau$,
  whenever $T(\sigma) = \tau$.
  We write $\sigma \fstred \tau$, and call $\tau$ a \emph{transduct} of $\sigma$,
  if there exists an FST $T$ such that $\sigma \fstred_{T} \tau$.
\end{definition}

Clearly, the relation~$\fstred$ is reflexive.
By composition of FSTs (the so-called `wreath product'), 
the relation $\fstred$ is also transitive, see~\cite[Remark~9]{endr:hend:klop:2011}.
We write $\sigma \fstredstrict \tau$ when $\sigma \fstred \tau$ but not $\tau \fstred \sigma$.
Whenever we have $\sigma \fstred \tau$ as well as a back-transduction $\tau \fstred \sigma$,
we consider $\sigma$ and $\tau$ to be equivalent.

\begin{definition}\label{def:fstconv}
  We define the relation~${\fstconv} \subseteq \str{\two} \times \str{\two}$ by 
  ${\fstconv} = {\fstred} \cap {\fstred^{-1}}$.
  For a sequence~$\sigma \in \str{\two}$
  the equivalence class $\convclass{\sigma} = \{\tau \where \sigma \fstconv \tau\}$
  is called the \emph{degree} of~$\sigma$.

  A degree $\convclass{\sigma}$ is an \emph{atom} if
  $\convclass{\sigma} \ne \botdeg$ and there is no degree 
  $\convclass{\tau}$ 
  such that $\convclass{\sigma} \fstredstrict \convclass{\tau} \fstredstrict \botdeg$.
\end{definition}

\section{Characterising Transducts of Spiralling Sequences}\label{sec:tools}

In this section we characterize the transducts of `spiralling' sequences.
Proofs omitted in the text can be found in the appendix.

\begin{definition}\label{def:seq}
  For a function $f : \nat \to \nat$ we define the sequence~$\seq{f} \in \str{\two}$
  by 
  \begin{center}
  $\displaystyle
    \seq{f} = \prod_{i = 0}^{\infty} 1 0^{f(i)} = 1 0^{f(0)} \, 1 0^{f(1)} \, 1 0^{f(2)} \, \cdots\,.
  $
  \end{center}
  For a sequence $\seq{f}$, we often speak of the $n$-th \emph{block} of $\seq{f}$ 
  to refer to the occurrence of the word $10^{f(n)}$ in $\seq{f}$.
\end{definition}

In the sequel we often write 
$\seq{f(n)]}$ 
to denote the sequence~$\seq{n \mapsto f(n)}$.
%
We note that there is a one-to-one correspondence between functions $f : \nat \to \nat$,
and infinite sequences over the alphabet~$\two$ that start with the letter~$1$ and that
contain infinitely many occurrences of the letter~$1$.
Every degree is of the form $\convclass{\seq{f}}$
for some $f : \nat \to \nat$.

The following lemma is concerned with some basic operations on functions
that have no effect on the degree of $\seq{f}$
(multiplication with a constant, and $x$- and $y$-shifts),
and others by which we may go to a lower degree
(taking subsequences, merging blocks).

\begin{lemma}\label{lem:basic}
  Let $f : \nat \to \nat$, and $a,b \in \nat$.
  It holds that:
  \begin{enumerate}
    \item 
      \label{item:lem:basic:mult}
      $\seq{a f(n)} \fstconv \seq{f(n)}$, for $a \gt 0$,
    \item 
      $\seq{f(n + a)} \fstconv \seq{f(n)}$,
      \label{item:lem:basic:xshift}
    \item 
      $\seq{f(n) + a} \fstconv \seq{f(n)}$,
      \label{item:lem:basic:yshift}
    \item 
      $\seq{f(n)} \fstred \seq{f(an)}$, for $a \gt 0$,
      \label{item:lem:basic:sub}
    \item 
      $\seq{f(n)} \fstred \seq{a f(2n) + b f(2n+1)}$.
      \label{item:lem:basic:merge}
  \end{enumerate}
\end{lemma}
\begin{proof}
  Recall that, for any $g : \nat \to \nat$,
  with the $n$-th \emph{block} of $\seq{g}$ we mean the factor $10^{g(n)}$ of $\seq{g}$.
  \begin{enumerate}
    \item 
      We have $\seq{a f(n)} \fstred_T \seq{f(n)}$ where $T$ is an FST that replaces 
      every factor $0^a$ by~$0$.
      Clearly, the inverse operation is also FST-realisable, so $\seq{f(n)} \fstred \seq{a f(n)}$.
    \item 
      An FST that prepends 
      $\prod_{i = 0}^{a-1} 1 0^{f(i)}$
      to the incoming word
      witnesses $\seq{f(n + a)} \fstred \seq{f(n)}$.
      For the converse direction, let the FST remove the first $a$ blocks from the input.
    \item 
      To see that $\seq{f(n) + a} \fstred \seq{f(n)}$, 
      consider an FST that removes a factor $0^a$ from every block in $\seq{f(n) + a}$.
      Conversely, $\seq{f(n)} \fstred \seq{f(n) + a}$ is witnessed by an FST that appends $0^a$ to every block.
    \item 
      It is clear that the operation of selecting every $a$-th block is realisable by an FST.
    \item 
      By Theorem~\ref{thm:transducts}.
    \qed
  \end{enumerate}
\end{proof}

\begin{definition}\label{def:periodic}
  Let $A$ be a set.
  A function $f : \nat \to A$ 
  is \emph{ultimately periodic}
  if for some integers $n_0 \ge 0$, $p > 0$ 
  we have
  $f(n + p) = f(n)$ for all 
  $n \ge n_0$.
\end{definition}  

\begin{definition}\label{def:spiralling}
  A function $f : \nat \to \nat$ 
  is called \emph{spiralling}
  if 
  \begin{enumerate}
    \item{}\label{def:spiralling:item:i}
      $\lim_{n \to \infty} f(n) = \infty$, and
    \item{}\label{def:spiralling:item:ii} 
      for every $m \ge 1$, the function $n \mapsto f(n) \bmod m$ is ultimately periodic.
  \end{enumerate}
\end{definition}

Functions with the property \ref{def:spiralling:item:ii} in Definition~\ref{def:spiralling}
have been called `ultimately periodic reducible' by Siefkes \cite{sief:1971} (quotation from \cite{seif:mcna:1976}). 
Note that the identity function is spiralling.
Furthermore scalar products, and pointwise sums and products of spiralling functions 
are again spiralling. 
As a consequence also polynomials $a_k n^k + a_{k-1} n^{k-1} + \cdots + a_0$ are spiralling.

In the remainder of this section, we will characterise the transducts $\sigma$ of $\seq{f}$ for spiralling $f$.
We will show that if such a transduct $\sigma$ is not ultimately periodic, 
then it is equivalent to a sequence $\seq{g}$ for a spiralling function $g$, 
and moreover, $g$ can be obtained from $f$ by a `weighted product'.

\begin{lemma}\label{lem:transducts}
  Let $f : \nat \to \nat$ be a spiralling function.
  We have $\seq{f} \fstred \sigma$
  if and only if 
  $\sigma$ is of the form
  \begin{align}
    \sigma = w \cdot \prod_{i = 0}^{\infty} \prod_{j = 0}^{m-1} p_j \, c_j^{\cyc{i}{j}}
    && \text{where} &&
    \cyc{i}{j} = \frac{f(n_0 + m i + j) - a_j}{z} \,,
    \label{eq:double:prod}
  \end{align}
  for some integers $n_0,m,a_j \ge 0$ and $z > 0$, 
  and finite words~$w, p_j, c_j \in \two^*$ $(0 \le j \lt m)$
  such that $\cyc{i}{j} \in \nat$ for all $i \in \nat$ and $j \in \nat_{<m}$.
\end{lemma}

\begin{proof}
  Assume $\seq{f} \fstred \sigma$ and let $T = \tup{Q,q_0,\delta,\lambda}$ be an FST
  that transduces $\seq{f}$ to $\sigma$.
  As 
  $f$ is spiralling, 
  there exist $\ell_0, p \in \nat$, $p > 0$ such that 
  $f(n) \equiv f(n+p) \pmod{\zloops{T}}$ for every $n \ge \ell_0$.
  Moreover, as a consequence of $\lim_{n \to \infty} f(n) = \infty$, there exists $\ell_1 \in \nat$ such that 
  $f(n) \ge |Q|$ for every $n \ge \ell_1$.

  For $n \in \nat$, let $q_n \in Q$ be the state that the automaton $T$ is in, before 
  reading the $n$-th occurrence of $1$ in the sequence $\sigma$
  (i.e., the start of the block $1 0^{f(n)}$). 
  By the pigeonhole principle there exist $n_0,m \in \nat$
  with $\max \{\ell_0,\ell_1\} < n_0$ and $m > 0$
  such that $m \equiv 0 \pmod{p}$ and $q_{n_0} = q_{n_0+m}$.
  Then, for every $i \in \nat$ and $j \in \nat_{<m}$, we have 
  $f(n_0 + mi + j) \ge |Q|$ and 
  $f(n_0 + mi + j) \equiv f(n_0 + j) \pmod{\zloops{t}}$.
  For $j \in \nat_{<m}$, 
  we define $a_j = \min \{f(n_0 + mi + j) \where i \in \nat\}$.
  Note that $a_j \ge |Q|$.
  Then for all $i \in \nat$ and $j \in \nat_{<m}$ we have: 
  $f(n_0 + mi + j) = a_j + \cyc{i}{j} \cdot z$
  where $\cyc{i}{j} = (f(n_0 + mi + j)-a_j) / z$ and $z = \zloops{T}$.
  Using Lemma~\ref{lem:pumping} it follows by induction that
  $q_{n_0+n} = q_{n_0+m+n}$ for every $n \in \nat$.
  Hence
  $q_{n_0+j} = q_{n_0+mi+j}$ for every $i \in \nat$ and $j \in \nat_{<m}$.
  Then, again by Lemma~\ref{lem:pumping}, for every $j \in \nat_{<m}$
  there exist words~$p_j,c_j \in \two^*$
  such that 
  $\lambda(q_{n_0+mi+j}, 1 0^{f(n_0 + mi + j)}) 
    = \lambda(q_{n_0+j}, 1 0^{a_j + \cyc{i}{j} \cdot \zloops{T}}) 
    = p_j c_j^{\cyc{i}{j}}$
  for all $i \in \nat$.
  We conclude that 
  \begin{align*}
    \sigma 
    = \prod_{n = 0}^{\infty} \lambda(q_n,10^{f(n)}) 
    = w \cdot \prod_{i = 0}^{\infty} \prod_{j = 0}^{m-1} p_j \, c_j^{\cyc{i}{j}}\,,
  \end{align*}
  where $w = \prod_{n = 0}^{n_0-1} \lambda(q_n,10^{f(n)})$.

  For the other direction, 
  assume $\sigma$ is of the form \eqref{eq:double:prod}
  for some $n_0,m,a_j,z \in \nat$ with $z > 0$, 
  and~$w, p_j, c_j \in \two^*$ $(0 \le j \lt m)$
  such that $\cyc{i}{j} = (f(n_0 + m i + j) - a_j)/z \in \nat$ for all $i \in \nat$ and $j \in \nat_{<m}$.
  We define an FST $T$ such that $\seq{f} \fstred_{T} \sigma'$, where 
  \begin{align*}
    \sigma' = \prod_{i = 0}^{\infty} \prod_{j = 0}^{m-1} p_j \, c_j^{\cyc{i}{j}}
  \end{align*}
  Define $T = \tup{Q,q_{m-1}^0,\delta,\lambda}$ 
  where $Q = \{q_j^h \where j \in \nat_{<m},\, {-}a_j \le h < z\}$,
  and 
  \begin{align*}
    \pair{\delta}{\lambda}(q_j^h,0) & = \pair{q_j^{h'}}{0^e} && 
    \text{where $e = 1$ and $h' = 0$ if $h + 1 = z$,} \\[-.5ex]
    &&&\text{and $e = 0$ and $h' = h + 1$, otherwise,} \\
    \pair{\delta}{\lambda}(q_j^h,1) & = \pair{q_{j'}^{{-}a_{j'}}}{1} && (j' = j+1 \bmod{m}) \,.
  \end{align*}
  For the verification of $\seq{f} \fstred_{T} \sigma'$ 
  we refer to the proof of Lemma~\ref{lem:wprod:FST}
  that contains a more general construction.
  The transduction $\sigma' \fstred \shift{\length{w}}{\sigma}$ 
  is realised by the FST in Figure~\ref{fig:easy}.
  Clearly also $\shift{\length{w}}{\sigma} \fstred \sigma$ holds, and we conclude $\seq{f} \fstred \sigma$.
  \qed
\end{proof}

\begin{figure}[t]
  \begin{center}
    \begin{tikzpicture}[scale=0.9,n/.style={circle,minimum size=8mm,draw,outer sep=1mm,inner sep=0mm}, node distance=25mm, l/.style={scale=.9}]
      \node (q0) [n] {$q_0$};
      \node (q1) [n,right of=q0] {$q_1$};
      \node (q2) [circle,minimum size=8mm,right of=q1,node distance=22mm] {$\cdots$};
      \node (qm-1) [n,right of=q2,,node distance=22mm] {$q_{m-1}$};
      \node (start) [right of=qm-1, node distance=11mm] {};
      \begin{scope}[->]
        \draw (start) -- (qm-1);
        \draw (q0) to[out=60,in=120,looseness=5] node [l,above] {$0|c_0$} (q0);
        \draw (q0) to node [l,above] {$1|p_1$} (q1);
        \draw (q1) to[out=60,in=120,looseness=5] node [l,above] {$0|c_1$} (q1);
        \draw (q1) to[dashed] node [l,above] {$1|p_2$} (q2);
        \draw (q2) to node [l,above] {$1|p_{m-1}$} (qm-1);
        \draw (qm-1) to[out=60,in=120,looseness=5] node [l,above] {$0|c_{m-1}$} (qm-1);
        \draw (qm-1) to[out=-140,in=-40,looseness=.5,pos=.45] node [l,above] {$1|p_0$} (q0);
      \end{scope}
    \end{tikzpicture}
  \end{center}
  \caption{%
    \textit{%
      An FST that transduces $\prod_{i = 0}^{\infty} \prod_{j = 0}^{m-1} 1 \, 0^{\cyc{i}{j}}$
      into $\prod_{i = 0}^{\infty} \prod_{j = 0}^{m-1} p_j \, c_j^{\cyc{i}{j}}$.%
  }}
  \label{fig:easy}
\end{figure}
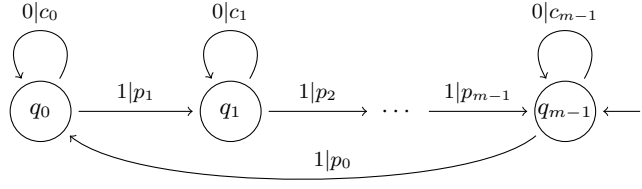

\begin{example}\label{ex:label}
  We illustrate the influence of both the zero-loops of the automaton 
  as well as the growth of the block lengths of the input sequence,
  on the size~$m$ of the inner product of Lemma~\ref{lem:transducts}.
  Consider the FST~$T = \tup{\{q_0,q_1,q_2\},q_0,\delta,\lambda}$ given by 
  \begin{center}
    \begin{tikzpicture}[n/.style={circle,minimum size=6mm,draw,outer sep=1mm}, node distance=25mm, l/.style={scale=.9}]
    \node (q0) [n] {$q_0$};
    \node (q1) [n,below right of=q0] {$q_1$};
    \node (q2) [n,above right of=q1] {$q_2$};
    \node (start) [left of=q0, node distance=10mm] {};
    
    \begin{scope}[->]
    \draw (q0) to[bend left=15] node [l,sloped,above,pos=.6,inner sep=.3mm] {$0|00$} (q1);
    \draw (q0) to[bend left=15] node [l,above] {$1|1$} (q2);
    
    \draw (q1) to[bend left=15] node [l,sloped,below] {$1|1$} (q0);
    \draw (q1) to[bend left=15] node [l,sloped,above,pos=.4,inner sep=.3mm] {$0|01$} (q2);
    
    \draw (q2) to[bend left=15] node [l,above] {$0|10$} (q0);
    \draw (q2) to[bend left=15] node [l,sloped,below] {$1|1$} (q1);
    
    \draw (start) -- (q0);
    \end{scope}
    \end{tikzpicture}
  \end{center}
  and the sequence 
  \begin{align*}
    \seq{\floor{\frac{n}{2}}} = 1 1 10 10 10^2 10^2 10^3 10^3 10^4 10^4 \cdots \,,
  \end{align*}
  where $\floor{x} = \max \{ n \in \nat \where n \le x \}$.
%
  We investigate the sequence $r_0 a_0 r_1 a_1 r_2 \cdots$ of states~$r$ of $T$ alternated 
  with letters~$a$ from the input sequence, in such a way that $T$ is in state $r_n$ 
  after having read the word $a_0 a_1 \cdots a_{n-1}$,
  \begin{align*}
    \underline{q_0} 1 q_2 1 q_1 1 q_0 0 q_1 1 q_0 0 
    q_1 1 q_0 0 q_1 0 q_2 1 q_1 0 q_2 0 
    \underline{q_0} 1 q_2 0 q_0 0 q_1 0 q_2 1 q_1 0 q_2 0 q_0 0 \cdots
  \end{align*}
  The two underlined occurrences of $q_0$ indicate a repetition of states 
  in combination with a repetition of the block size modulo $\zloops{T} = 3$.
  The number of blocks between these occurrences, forming the repetition, is $m = 6$.
  Actually, following the proof of Lemma~\ref{lem:transducts} precisely, 
  the algorithm would instead select the repetition 
  starting from the second underlined occurrence of $q_0$. 
  The reason is that in general, when reading a block $100\cdots0$,
  only after reading $|Q|$ zeros, we are guaranteed to be in a zero-loop.
  For this FST~$T$, all states are on a zero-loop, 
  and so we enter the loop immediately.
\end{example}

From Lemma~\ref{lem:transducts} it follows that FSTs can only transform 
the length of blocks by linear functions (and merge blocks).
As a consequence, FSTs typically cannot slow down the growth rate of blocks 
in spiralling sequences by more than a linear factor.
This yields the following simple criterion for non-reducibility. 
\begin{lemma}\label{lem:too:fast}
  Let $f,g : \nat \to \nat$ be
  such that $f$ is spiralling, $g$ is not ultimately periodic and $g \in o(f)$, 
  i.e., for every $a \in \nat$ there is $n_0 \in \nat$
  such that for all $n \ge n_0$ it holds that $f(n) \ge a g(n)$.
  Then $\seq{f} \not\fstred \seq{g}$. 
\end{lemma}
Note that the reverse $\seq{g} \not\fstred \seq{f}$ does \emph{not} follow.
For example, we have $\seq{4^n} \not\fstred \seq{2^n}$, but $\seq{2^n} \fstred \seq{4^n}$.

There can be several ways to factor the transduct~$\sigma$ 
in the statement of Lemma~\ref{lem:transducts},
as shown by the following example.

\begin{example}\label{ex:ambigue}
  Consider the following FST~$T = \tup{\{q_0,q_1\},q_0,\delta,\lambda}$
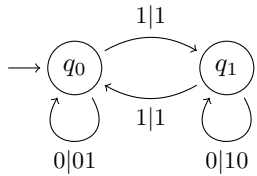
  \begin{center}
    \begin{tikzpicture}[n/.style={circle,minimum size=6mm,draw,outer sep=1mm}, node distance=20mm, l/.style={scale=.9}]
      \node (q0) [n] {$q_0$};
      \node (q1) [n,right of=q0] {$q_1$};
      \node (start) [left of=q0, node distance=10mm] {};
      
      \begin{scope}[->]
      \draw (q0) to[bend left=30] node [l,above] {$1|1$} (q1);
      \draw (q0) to[out=-60,in=-120,looseness=5] node [l,below] {$0|01$} (q0);
  
      \draw (q1) to[bend left=30] node [l,below] {$1|1$} (q0);
      \draw (q1) to[out=-60,in=-120,looseness=5] node [l,below] {$0|10$} (q1);
  
      \draw (start) -- (q0);
      \end{scope}
    \end{tikzpicture}
  \end{center}
  together with the sequence $\seq{n} = 1 10 10^2 10^3 10^4 \cdots$.
  The $T$-transduct of $\seq{n}$ is 
  \begin{align*}
    T(\seq{n}) = 1 \, 1 (01) \, 1 (10)^2 \, 1 (01)^3 \, 1 (10)^4 \, \cdots\,.
  \end{align*}
  Using the double product~\eqref{eq:double:prod} of Lemma~\ref{lem:transducts} 
  we have $w = \emptyword$,
  $n_0 = 0$, $m = 2$, $p_0 = p_1 = 1$, $c_0 = 01$, $c_1 = 10$, $a_0 = a_1 = 0$, and $z = \zloops{T} = 1$.
  Thus, $\cyc{i}{j} = 2i + j$.  
  Note that $c_1^\omega = p_0 c_0^\omega$, and so we can merge the factors of the inner product, and decrease its size to $m = 1$.
  Then $T(\seq{n}) = 1101 \, 1101 (0101) \, 1101 (0101)^2 \, 1101 (0101)^3 \, \cdots$, that is, 
  in the double product~\eqref{eq:double:prod}, $T(\seq{n})$ can be factored by choosing 
  $m = 1$, $p_0 = 1101$, and $c_0 = 0101$.
\end{example}

Whenever we have the situation $c_j^\omega = p_{j+1}c_{j+1}^\omega$ in the representation~\eqref{eq:double:prod}
for some $j \in \nat_{<m}$ (addition is modulo~$m$), we speak of a `transition ambiguity'.
We will eliminate such ambiguities by merging factors of the inner product, 
as in~Lemma~\ref{lem:merge}.
In general, this merging may involve weighted sums in the exponents. 
The following lemma is folklore, see for instance~\cite{caut:mign:shal:wang:yazd:2003}.
\begin{lemma}\label{lem:folk}
  If a sequence $\sigma$ is periodic with period lengths $k$ and $\ell$,
  then it is periodic with period length $\gcd(k,\ell)$. 
\end{lemma}

\begin{lemma}\label{lem:merge}
  Let $u,v,w \in \two^*$ be finite words with $u,w \ne \emptyword$ 
  such that $u^\omega = v w^\omega$.
  Then there exists a word $x \in \two^*$ and $a,b \in \nat$ such that
  $u^m vw^n = vx^{am+bn}$ for all $m,n \in \nat$.  
\end{lemma}
\begin{proof}
  Let $d = \gcd(|u|,|w|)$, $a = |u|/d$ and $b = |w|/d$.
  Let $x$ be the suffix of~$w$ of length $d$.
  We argue that this choice has the desired property. 
  %
  The sequence $\sigma = u^\omega = v w^\omega$ is periodic with period lengths $|u|$ and $|w|$.
  By Lemma~\ref{lem:folk}
  it follows that $\sigma$ is also periodic with period length~$d$.
  Hence $w = x^b$ and 
  for all $m,n \in \nat$, $u^m v w^n$ 
  is a prefix of $\sigma$.
  Likewise $v x^{am + bn}$ is a prefix of~$\sigma = v x^\omega$,
  and both prefixes have the same length:
  $|u^m vw^n| = m|u| + |v| + n|w| = |v| + am|x| + bn|x| = |v x^{am + bn}|$.
  We conclude $u^m vw^n = v x^{am + bn}$.
  \qed
\end{proof}

Our goal is to obtain a simple characterisation of the degrees of the transducts~$\sigma$ 
of spiralling sequences~$\seq{f}$.
To this end, we transform the transduct~$\sigma$ into a sequence~$\sigma'$
by replacing in the double product of Lemma~\ref{lem:transducts}
the displayed occurrence of~$p_j$ by~$1$ and of~$c_j$ by~$0$ 
for every $j \in \nat_{<m}$.
To guarantee that this transformation does not change the degree, 
that is $\sigma' \fstconv \sigma$,
we first have to resolve transition ambiguities.

For the back transformation blocks $10^{\cyc{i}{j}}$ 
have to be replaced by $p_j c_j^{\cyc{i}{j}}$,
an operation that is easily realised by an FST,
see Figure~\ref{fig:easy}.

If the product does not contain transition ambiguities, 
then also the transformation from $\sigma$ into $\sigma'$ 
can be realised by an FST and thus does not change the degree of the sequence, 
hence $\sigma \fstconv \sigma'$.
If there exists $j \in \nat_{<m}$ with
$c_j^\omega = p_{j+1}c_{j+1}^\omega$,
then, by this transformation of $\sigma$ into $\sigma'$, 
one possibly leaves the degree of $\sigma$,
i.e., $\sigma' \fstredstrict \sigma$.
because for large enough~$i\in\nat$, 
an FST cannot recognise where a block~$p_j c_j^{\cyc{i}{j}}$ ends 
and where the next block~$p_{j+1} c_{j+1}^{\cyc{i}{j+1}}$ starts.
This might make it impossible to realise the transformation by an FST,
as then the FST cannot replace $p_{j+1}$ by $1$.

\begin{definition}\label{def:weighted:product}
  A \emph{weight} is a tuple $\tup{a_0,\ldots,a_{k-1},b} \in \rat^{k+1}$ 
  of rational numbers
  such that $a_0,\ldots,a_{k-1} \ge 0$. 
  Given a weight~$\alpha = \tup{a_0,\ldots,a_{k-1},b}$ and a function $f : \nat \to \nat$
  we define
  $\wof{\alpha}{f} \in \rat$ by
  \begin{align*}
    \wof{\alpha}{f} \;=\; a_0 f(0) + a_1 f(1) + \cdots + a_{k-1} f(k-1) + b \,.
  \end{align*}
  The weight~$\alpha$ is called \emph{constant} when $a_j = 0$ for all $j \in \nat_{<k}$.
  For a tuple of weights $\vec{\alpha} = \tup{\alpha_0,\ldots,\alpha_{m-1}}$ 
  we define its \emph{rotation} by $\vec{\alpha}' = \tup{\alpha_1,\ldots ,\alpha_{m-1},\alpha_0}$.
  
  
  For functions $f : \nat \to \nat$,
  and tuples $\vec{\alpha} = \tup{\alpha_0,\alpha_1,\ldots, \alpha_{m-1}}$ of weights,
  the \emph{weighted product} of $\vec{\alpha}$ and $f$ is a function $\wprod{\vec{\alpha}}{f} : \nat \to \rat$
  that is defined by induction on $n$ through the following scheme of equations: 
  \begin{align*}
    (\wprod{\vec{\alpha}}{f})(0) & = \alpha_0 \cdot f \\
    (\wprod{\vec{\alpha}}{f})(n+1) & = (\wprod{\vec{\alpha}'}{\shift{|\alpha_0|-1}{f}})(n) && (n \in \nat)
  \end{align*}
  where $|\alpha_i|$ is the length of the tuple $\alpha_i$,
  and $\shift{k}{f} $ is the $k$-th shift of $f$.
  A weighted product $\wprod{\vec{\alpha}}{f}$ is called \emph{natural} 
  if $(\wprod{\vec{\alpha}}{f})(n) \in \nat$ for all $n \in \nat$. 
\end{definition}

In what follows, all weighted products $\wprod{\vec{\alpha}}{f}$ that we consider are assumed to be natural.

\begin{example}
  Let $f(n) = n$ for all $n \in \nat$, 
  and $\vec{\alpha} = \tup{\alpha_1,\alpha_2}$
  with $\alpha_1 = \tup{1,2,3,4}$, $\alpha_2 = \tup{0,1,1}$.
  Interpreting the functions $f$ and $\wprod{\vec{\alpha}}{f}$ as sequences,
  the computation of $\wprod{\vec{\alpha}}{f}$ can be visualised as follows:
  \begin{center}
    \begin{tikzpicture}
      \node at (-.8cm,0) [anchor=east] {$f$};
      \node at (11*.9cm,0) [anchor=east] {$\cdots$};
      \foreach \i in {0,1,2,3,4,5,6,7,8,9} {
        \node (\i) at (\i*.9cm,0) {\i};
      }
      \node at (-.8cm,-1.2cm) [anchor=east] {$\wprod{\vec{\alpha}}{f}$};
      \node at (11*.9cm,-1.2cm) [anchor=east] {$\cdots$};
      \foreach \i/\x/\v in {0/1/12,1/3.5/5,2/6/42,3/8.5/10} {
        \node (s\i) at (\x*.9cm,-1.2cm) {\v};
      }
      \begin{scope}[inner sep=0,->,nodes={scale=.8}]
      \draw (0) -- node [xshift=3mm,pos=.3] {$\times1$} (s0); 
      \draw (1) -- node [xshift=2.5mm,pos=.3] {$\times2$} (s0); 
      \draw (2) -- node [xshift=3mm,pos=.3] {$\times3$} node [xshift=3mm,pos=.8] {$+4$} (s0); 
      \draw (3) -- node [xshift=3mm,pos=.3] {$\times0$} (s1); 
      \draw (4) -- node [xshift=3mm,pos=.3] {$\times1$} node [xshift=3mm,pos=.8] {$+1$} (s1); 
      \draw (5) -- node [xshift=3mm,pos=.3] {$\times1$} (s2); 
      \draw (6) -- node [xshift=2.5mm,pos=.3] {$\times2$} (s2); 
      \draw (7) -- node [xshift=3mm,pos=.3] {$\times3$} node [xshift=3mm,pos=.8] {$+4$} (s2); 
      \draw (8) -- node [xshift=3mm,pos=.3] {$\times0$} (s3); 
      \draw (9) -- node [xshift=3mm,pos=.3] {$\times1$} node [xshift=3mm,pos=.8] {$+1$} (s3); 
      \end{scope}
    \end{tikzpicture}
  \end{center}
  Thus, for $n = 0,1,2,3\ldots$\,, $(\wprod{\vec{\alpha}}{f})(n)$ takes the values $12,5,42,10,\ldots$\,. 
\end{example}

\begin{lemma}\label{lem:wprod:nth}
  Let $\vec{\alpha}$ be an $m$-tuple of weights ($m > 0$), and let $f : \nat \to \nat$.
  For all $n \in \nat$ we have 
  $(\wprod{\vec{\alpha}}{f})(n) = \wof{\alpha_r}{\shift{t}{f} }$
  where $q,r \in \nat$ with $r < m$ are such that $n = qm + r$,
  and $t = q \cdot \sum_{j = 0}^{m-1}(\length{\alpha_j} - 1) + \sum_{j = 0}^{r-1}(\length{\alpha_j} - 1)$.
\end{lemma}
\begin{proof}
  By induction on $n \in \nat$ we show that, 
  for all $f : \nat \to \nat$, and all $m$-tuples of weights~$\vec{\alpha}$, 
  we have
  \begin{align*}
    (\wprod{\vec{\alpha}}{f})(n) = \wof{\alpha_r}{\shift{t}{f}}
    &&&\text{where $n = qm + r$, and $q,r \in \nat$ with $r < m$,}\\
    &&&\text{and $t = q \cdot \sum_{i = 0}^{m-1}(\length{\alpha_i} - 1) + \sum_{i = 0}^{r-1}(\length{\alpha_i} - 1)$.}
  \end{align*}
  The base case, $n = 0$ follows directly from the definition.
  For the induction step let $\vec{\alpha}$ and $f : \nat \to \nat$ be arbitrary,
  let $n = n' + 1$, $q' = \floor{n'/m}$, and $r' = n' - q'm$.
  We abbreviate $k_i = \length{\alpha_i} - 1$ ($i \in \nat_{<m}$) 
  and $s_\ell = \sum_{i=0}^{\ell-1}$ ($\ell \in \nat_{\le m}$).
  By Definition~\ref{def:weighted:product} and the induction hypothesis we obtain
  \begin{align*}
    (\wprod{\vec{\alpha}}{f})(n) 
    = (\wprod{\vec{\alpha}'}{\shift{k_0}{f}})(n') 
    = \wof{\alpha_{r'+1}}{\shift{t'}{\shift{k_0}{f}}}
  \end{align*}
  with $t' = q' \cdot s_m + s_{r'+1} - k_0$, 
  and where addition in the subscript of $\alpha$ is modulo~$m$.
  We conclude $(\wprod{\vec{\alpha}}{f})(n) = \wof{\alpha_{r}}{\shift{t}{f}}$
  where $n = (q'+1)m$, $r = 0$ and $t = (q'+1) \cdot s_m$ if $r' = m - 1$,
  and $n = q'm + r$, $r = r' + 1$ and $t = q' \cdot s_m + r$ if $r' < m - 1$.
\end{proof}

\begin{lemma}\label{lem:wprod:not:botdeg}
  Let $f : \nat \to \nat$.
  If $\seq{\wprod{\vec{\alpha}}{f}} \not\in \botdeg$, 
  then there exists $i \in \nat_{<\length{\vec{\alpha}}}$ such that $\alpha_i$ is a non-constant weight.
\end{lemma}
\begin{proof}
  Assume $\alpha_i$ is a constant weight for all $i \in \nat_{<\length{\vec{\alpha}}}$; 
  let $\alpha_i = \tup{0,0,\ldots,0,b_i}$. 
  Then 
  $\wprod{\vec{\alpha}}{f} : \nat \to \nat$ is periodic: for all $n\in\nat$, 
  we have $(\wprod{\vec{\alpha}}{f})(n) = b_i$ where $i \equiv n \pmod{\length{\vec{\alpha}}}$
  and $i \in \nat_{<\length{\alpha}}$. Hence $\seq{\wprod{\vec{\alpha}}{f}} \in \botdeg$.
\end{proof}

We define the operation~$\zip$~\cite{grab:endr:hend:klop:moss:2012} also known as `perfect shuffle'.

\begin{definition}\label{def:zip}
  Let $k \in \nat$. 
  For $i \in \nat_{<k}$, 
  let $f_i : \nat \to \nat$ be a function.
  We define the function $\zip_k(f_0,f_1,\ldots,f_{k-1}) : \nat \to \nat$ by 
  \begin{align*}
    \zip_k(f_0,f_1,\ldots,f_{k-1})(kn+i) = f_i(n) && (n \in \nat,\, i \in \nat_{<k}) \,.
  \end{align*}
\end{definition}

\begin{lemma}\label{lem:zip:periodic:modulo}
  Let $f_0,f_1,\ldots,f_{k-1} : \nat \to \nat$ be 
  such that
  $f_i$ ($i \in \nat_{<k}$) is ultimately periodic modulo every $m \ge 1$ (see Definition~\ref{def:spiralling}~(ii)).
  Then the function $\zip_k(f_0,f_1,\ldots,f_{k-1})$ is ultimately periodic modulo every $m \ge 1$.
\end{lemma}
\begin{proof}
  Fix an arbitrary $m \ge 1$.
  From the assumption we obtain, for every $i \in \nat_{<k}$, the existence of $n_i, p_i \in \nat$ with $p_i \ge 1$
  such that $f_i(n+p_i) \equiv f_i(n) \pmod{m}$ for all $n \ge n_0$.
  Define $n' = \max\{n_i\where i \in \nat_{<k}\}$
  and $p = \mrm{lcm} \{p_i\where i \in \nat_{<k}\}$.
  Let $f = \zip_k(f_0,f_1,\ldots,f_{k-1})$.
  Let $n \ge n'$ be arbitrary, 
  and let $q,i\in\nat$ with $i < k$ be such that $n = kq + i$.
  Then 
  we have 
  $f(n+kp) = f(k(q+p) + i) = f_i(q+p) \equiv f_i(q) \pmod{m}$,
  and 
  $f(n) = f(kq + i) = f_i(q)$, as desired.
  \qed
\end{proof}

\begin{lemma}\label{lem:wprod:zip}
  Let $f : \nat\to\nat$ be a function, 
  and let $\vec{\alpha}$ be an $m$-tuple of weights.
  Then there is a $m$-tuple of weights $\vec{\beta}$
  such that
  \begin{align*}
    \wprod{\vec{\alpha}}{f} & = \zip_m(g_0,g_1,\ldots,g_{m-1}) && \text{where $g_i = \wprod{\tup{\beta_i}}{f}$ for $i \in \nat_{<m}$.}
  \end{align*}
\end{lemma}
\begin{proof}
  Let $k_i = \length{\alpha_i}-1$ ($i \in \nat_{<m}$), 
  and $\alpha_i = \tup{a_{i,0},a_{i,1},\ldots,a_{i,k_i-1},b_i}$.
  For $\ell \in \nat_{\le m}$ we define $s_\ell = \sum_{i=0}^{\ell - 1} k_i$.
  For $i \in \nat_{<m}$, we define the weight $\beta_i$ of length 
  $s_m + 1$ by
  \begin{align*}
    (\beta_i)_{s_m} = b_i
    &&
    (\beta_i)_{s_j + h} =
    \begin{cases}
      a_{i,h} & \text{if $j = i$,} \\  
      0 & \text{if $j \ne i$}  
    \end{cases}
    && (j \in \nat_{<m},\, h \in \nat_{<k_i})\,.
  \end{align*}
  Let $g = \zip_m(g_0,g_1,\ldots,g_{m-1})$.
  Now it is just a matter of unfolding definitions to derive  that 
  $\wprod{\vec{\alpha}}{f} = g$.
  Fix an arbitrary $n \in \nat$, 
  and let $k,i\in\nat$ with $i < m$ be such that $n = km + i$.
  Then we have, using Lemma~\ref{lem:wprod:nth} twice, 
  $(\wprod{\vec{\alpha}}{f})(n) 
  = (\wof{\alpha_i}{\shift{k s_m + s_i}{f}})(n) 
  = a_{i,0} \cdot f(k s_m + s_i) + a_{i,1} \cdot f(k s_m + s_i + 1) + \cdots + a_{i,k_i-1} \cdot f(k s_m + s_i + k_i - 1) + b_i
  = (\beta_i)_{0} \cdot f(k s_m) + \cdots + (\beta_i)_{s_m - 1} \cdot f(k s_m + s_m - 1) + (\beta_i)_{s_m}
  = (\wof{\beta_i}{\shift{k s_m}{f}})(k)
  = (\wprod{\tup{\beta_i}}{f})(k) = g_i(k) = g(km+i) = g(n)$.
  \qed
\end{proof}

\begin{lemma}\label{lem:wprod:preserve:spiralling}
  Let $f:\nat\to\nat$ be spiralling, and let $\vec{\alpha}$ be a tuple of non-constant weights.
  Then $\wprod{\vec{\alpha}}{f}$ is spiralling.
\end{lemma}
\begin{proof}
  Let $f:\nat\to\nat$ be spiralling, and let $\vec{\alpha}$ be an $r$-tuple of non-constant weights.
  It is easy to see that 
  $\lim_{n \to \infty}(\wprod{\vec{\alpha}}{f})(n) = \infty$.

  To see that $g = \wprod{\vec{\alpha}}{f}$ is ultimately periodic modulo every $m \ge 1$,
  i.e., $\forall m\,.\exists n_0,p.\,\forall n \ge n_0\,. g(n) \equiv g(n+p) \pmod{m}$,
  we first we prove the claim for $r = 1$. 
  So $\vec{\alpha} = \tup{\alpha}$ is a tuple consisting of one non-constant weight.
  Without loss of generality, let $\alpha = \tup{a_0/d,a_1/d,\ldots,a_{k-1}/d,b/d}$, for some integer~$d > 0$.
  Let $g = \wprod{\tup{\alpha}}{f}$. 
  By Lemma~\ref{lem:wprod:nth} we have
  $g(n) = \sum_{j=0}^{k-1} \big((a_j/d) \cdot f(kn+j)\big) + b/d$, for all $n \in \nat$.
  Fix an arbitrary integer $m > 0$.
  By $f$ being spiralling there exist
  $n_0,p\in\nat$ with $p > 0$
  such that $f(n+p) \equiv f(n) \pmod{dm}$ for all $n \ge n_0$.
  This implies that $(a_j/d) \cdot f(k(n+p)+j) \equiv (a_j/d) \cdot f(kn+j) \pmod{m}$ for all $n \ge n_0$.
  Hence we obtain $g(n+p) \equiv g(n) \pmod{m}$.

  For the case $r > 1$ we reason as follows.
  By Lemma~\ref{lem:wprod:zip} there exists an $r$-tuple $\vec{\beta}$ such that 
  $g = \wprod{\vec{\alpha}}{f} = \zip_r(g_0,g_1,\ldots,g_{m-1})$,
  where $g_i = \wprod{\tup{\beta_i}}{f}$ ($i \in \nat_{<r}$).
  By the above argument (for $r = 1$), 
  we have that $g_i$ is ultimately periodic modulo every $m \ge 1$.
  We conclude by Lemma~\ref{lem:zip:periodic:modulo}.
  \qed
\end{proof}

We will show that weighted products give rise to a characterisation, up to equivalence~$\fstconv$, 
of functions realised by FSTs on the set of spiralling sequences, see Theorem~\ref{thm:transducts}.

\begin{lemma}\label{lem:wprod:FST}
  Let $f : \nat \to \nat$, and 
  $\vec{\alpha}$ a tuple of weights.
  If ${\wprod{\vec{\alpha}}{f}}$ is a natural weighted product, then we have
  $\seq{f} \fstred \seq{\wprod{\vec{\alpha}}{f}}$.
\end{lemma}
\begin{proof}
  Let $m = \length{\vec{\alpha}}$, 
  $k_i = \length{\alpha_i} - 1$,
  and $\alpha_i = \tup{a_{i,0}/d_i,a_{i,1}/d_i,\ldots,$ $a_{i,k_i-1}/d_i,b_i/d_i}$
  with $a_{i,j},d_i \in \nat$ and $b_i \in \mbb{Z}$ for all $i \in \nat_{\lt m}$ and $j \in \nat_{<k_i}$
  (clearly weights can always be brought into this form).

  We define $T = \tup{Q,q_{m-1,k_{m-1}-1}^{0},\delta,\lambda}$
  consisting of states $q_{i,j}^{h}$ for every $i \in \nat_{<m}$, $j \in \nat_{<k_i}$, 
  and $h$ such that $\min(0,b_i) \le h \lt d_i$.
  The superscript $h$ ($\min(0,b)_i \le h \lt d_i$) in a state $q_{i,j}^h$
  indicates the amount $h/d_i$ of zeros that still has to be consumed/produced.
  The transition and output functions
  $\pair{\delta}{\lambda} : Q \times \two \to Q \times \two^*$ of $T$ 
  are defined as follows; 
  let $\sidifnn : \zz \to \nat$ be defined by $\idifnn{z} = z$ if $z \ge 0$ and $\idifnn{z} = 0$ otherwise.
  \begin{align*}
    \pair{\delta}{\lambda}(q_{i,j}^{h},0) & = \pair{q_{i,j}^{h'}}{0^e} 
    && \text{where $e = \floor{\idifnn{h + a_{i,j}}/d_i}$, $h' = h + a_{i,j} - e d_i$}\\
    \pair{\delta}{\lambda}(q_{i,j}^{h},1) & = \pair{q_{i,j+1}^{h}}{\emptyword}
    && (j < k_i - 1) \\
    \pair{\delta}{\lambda}(q_{i,k_i-1}^{h},1) & = \pair{q_{i',0}^{h'}}{1 0^e} 
    && \text{where $e = \floor{\idifnn{b_{i'}}/d_{i'}}$, $h' = b_{i'} - e d_{i'}$\,,}
  \end{align*}
  where 
  $i' = i+1 \bmod{m}$.
%

  We now show $\seq{f} \fstred_T \seq{\wprod{\vec{\alpha}}{f}}$.
  For $i \in \nat_{\le m}$, define $s_i = k_0 + k_1 + \cdots + k_{i-1}$. 
  We fix arbitrary $i \in \nat_{<m}$ and $n \in \nat$ with $n \equiv s_i \pmod{s_m}$.
  After reading the prefix $10^{f(0)} \cdots 10^{f(n-1)} 1$ of $\seq{f}$,
  $T$ is in state $q_{i,0}^{b_i \bmod{d_i}}$.
  Let $u$ denote the prefix of $\shift{}{\seq{\shift{n}{f}}}$  
  of the form $u = 0^{f(n)} 1 0^{f(n+1)} \cdots 1 0^{f(n+k_i-1)} 1$.
  We show that 
  \begin{align}
    \pair{\delta}{\lambda}(q_{i,0}^{b_i - e d_i},u) = \pair{q_{i',0}^{b_{i'} - e' d_{i'}}}{0^{\wof{\alpha_i}{\shift{n}{f}}-e} 1 0^{e'}}
    \label{eq:lem:wprod:FST:i}
  \end{align}
  where $i' = i+1 \bmod{m}$, 
  $e = \floor{\idifnn{b_{i}}/d_{i}}$, and
  $e' = \floor{\idifnn{b_{i'}}/d_{i'}}$.
  By the (implicit) assumption that $(\wprod{\vec{\alpha}}{f})(n) \in \nat$ for all $n \in \nat$,
  we have $\wof{\alpha_i}{\shift{n}{f} } \in \nat$.
  Also note that
  $\wof{\alpha_i}{\shift{n}{f}} = (b_i + \sum_{j=0}^{k_i-1}a_{i,j} \cdot f(n + j))/d_i$.

  For $j \in \nat_{< k_i}$, 
  we define $h_j, e_j \in \nat$ with $\min(0,b_i) \le h_j \lt d_i$ as follows:
  \begin{align*}
      h_0 & = b_i - e d_i \\
      h_{j+1} & = h_j + a_{i,j} \cdot f(n+j)  - e_j \cdot d_i && (j \lt k_i - 1) \\
      e_{j} & = \floor{\frac{h_j + a_{i,j} \cdot f(n+j)}{d_i}} \,.
  \end{align*}
  Then we have 
  \begin{align*}
    \pair{\delta}{\lambda}(q_{i,j}^{h_j},0^{f(n+j)} 1) & = \pair{q_{i,j+1}^{h_{j+1}}}{0^{e_{j}}}
    && \text{for $j < k_i - 1$, and} \\
    \pair{\delta}{\lambda}(q_{i,k_i-1}^{h_{k_i-1}},0^{f(n+k_i-1)} 1) & = \pair{q_{i',0}^{b_{i'}}}{0^{e_{k_i-1}} 1}
    && \text{where $i' = i+1 \bmod{m}$.}
  \end{align*}
  Moreover, as we have
  \begin{align*}
    \sum_{j = 0}^{k_i-1} e_j 
    = \floor{\frac{h_0 + \sum_{j=0}^{k_i-1}a_{i,j} \cdot f(n+j)}{d_i}} 
    = \wof{\alpha}{\shift{n}{f}} - e \,,
  \end{align*}
  we conclude that \eqref{eq:lem:wprod:FST:i} holds.
  \qed
\end{proof}

\begin{definition}\label{def:doubleprod}
  Let $f : \nat \to \nat$ be a function, and, for some $m > 0$, let $\vec{\alpha}$ be an $m$-tuple of weights
  such that $\wprod{\vec{\alpha}}{f}$ is a natural weighted product.
  Let $\vec{p}$ and~$\vec{c}$\, be $m$-tuples of finite words. 
  We define the sequence $\dpf{f,\vec{\alpha},\vec{p},\vec{c}} \in \str{\two}$ by
  \begin{align*}
    \dpf{f,\vec{\alpha},\vec{p},\vec{c}}
    = \prod_{i = 0}^{\infty} \prod_{j = 0}^{m-1} p_j \, c_j^{\cyc{i}{j}}
    && \text{where} &&
    \cyc{i}{j} = (\wprod{\vec{\alpha}}{f})(mi + j) \,.
  \end{align*}
\end{definition}

Note that $\seq{f}$ can also be cast into this notation:
$\seq{f} = \dpf{f,\tup{\tup{1,0}},\tup{1},\tup{0}}$.
For the following lemma we recall that 
for a tuple $\vec{a} = \tup{a_0,a_1,\ldots,a_{k-1}}$,
we write $\vec{a}'$ for the rotation~$\tup{a_1,\ldots,a_{k-1},a_0}$.

\begin{lemma}\label{lem:doubleprod:rotation}
  Let $f$, $\vec{\alpha}$, $\vec{p}$, 
  $\vec{c}$ be as in Definition~\ref{def:doubleprod}. 
  We have
  $\dpf{f,\vec{\alpha},\vec{p},\vec{c}} = 
  p_0 c_0^{\wof{\alpha_0}{f}} \cdot \dpf{\shift{\length{\alpha_0}-1}{f},\vec{\alpha}',\vec{p}',\vec{c}'}$.
\end{lemma}
\begin{proof}
  This is a straightforward calculation.
  Let $\length{\vec{\alpha}} = \length{\vec{p}} = \length{\vec{c}} = m$,
  and let $\cyc{i}{j} = (\wprod{\vec{\alpha}}{f})(mi+j)$.
  For $i \in \nat$ and $j \in \nat_{<m}$ we define $\cycp{i}{j}$ and $\psi(i,j)$ by
  \begin{align*}
    \cycp{i}{j} &= 
    \begin{cases}
      \cyc{i}{j+1} & \text{if $j < m-1$,} \\
      \cyc{i+1}{0} & \text{if $j = m-1$,}
    \end{cases}
    \\
    \psi(i,j) & = (\wprod{\vec{\alpha}'}{\shift{\length{\alpha_0} - 1}{f}})(mi+j)
  \end{align*}
  First we note that $\cycp{i}{j} = \psi(i,j)$ for all $i \in \nat$ and $j \in \nat_{<m}$:
  Let $i \in \nat$ and $j \in \nat_{<m}$ be arbitrary.
  If $j < m - 1$, then 
  $\cycp{i}{j} = \cyc{i}{j+1} = (\wprod{\vec{\alpha}}{f})(mi+j+1) 
  = (\wprod{\vec{\alpha}'}{\shift{\length{\alpha_0} - 1}{f}})(mi+j) = \psi(i,j)$ by Definition~\ref{def:weighted:product}.
  For the case $j = m - 1$, we find 
  $\cycp{i}{m-1} = \cyc{i+1}{0} = (\wprod{\vec{\alpha}}{f})(m(i+1)) 
  = (\wprod{\vec{\alpha}'}{\shift{\length{\alpha_0} - 1}{f}})(mi+m-1) = \psi(i,m-1)$, 
  again by definition of weighted products. 
  Then we have
  \begin{align*}
    \dpf{f,\vec{\alpha},\vec{p},\vec{c}} 
    & = \prod_{i=0}^\infty \prod_{j=0}^{m-1} p_j c_j^{\cyc{i}{j}} \\
    & = p_0 c_0^{\cyc{0}{0}} \cdot \prod_{i=0}^\infty \prod_{j=0}^{m-1} p_{j+1} c_{j+1}^{\cycp{i}{j}} \\
    & = p_0 c_0^{\wof{\alpha_0}{f}} \cdot \prod_{i=0}^\infty \prod_{j=0}^{m-1} p_{j+1} c_{j+1}^{\psi(i,j)} \\
    & = p_0 c_0^{\wof{\alpha_0}{f}} \cdot \dpf{\shift{\length{\alpha_0}-1}{f},\vec{\alpha}',\vec{p}',\vec{c}'}
  \end{align*}
  where addition in the subscripts is computed modulo~$m$.
  \qed
\end{proof}

\begin{lemma}\label{lem:disambiguate}
  Let $f : \nat \to \nat$ be a spiralling function, 
  and let $\sigma \in \str{\two}$ be 
  such that $\seq{f} \fstred \sigma$ and $\sigma \not\in\botdeg$.
  Then there exist 
  $n_0, m \in \nat$, a word $w \in \two^*$,
  a tuple of weights~$\vec{\alpha}$,
  and tuples of words $\vec{p}$ and $\vec{c}$ 
  with $\length{\vec{\alpha}} = \length{\vec{p}} = \length{\vec{c}} = m > 0$
  such that:
  \begin{enumerate}
    \item 
      $\sigma = w \cdot \dpf{\shift{n_0}{f},\vec{\alpha},\vec{p},\vec{c}}$,
      \label{item:form}
    \item 
      $c_j^\omega \ne p_{j+1} c_{j+1}^\omega$ for every $j$ with $0 \le j < m-1$, 
      and $c_{m-1}^\omega \ne p_{0} c_{0}^\omega$, and
      \label{item:no:ambiguities}
    \item 
      $c_j \ne \emptyword$,
      and $\alpha_j$ is non-constant, for all $j \in \nat_{<m}$.
      \label{item:no:empty:cycles}
  \end{enumerate}
  
\end{lemma}
\begin{proof}
  By Lemma~\ref{lem:transducts}, 
  there exist $n_0, m, a_j, z \in \nat$ ($j \in \nat_{\lt m}$), 
  $w \in \two^*$, and $\vec{p}, \vec{c} \in (\two^*)^m$ such that 
  $\sigma = w \cdot \dpf{\shift{n_0}{f},\vec{\alpha},\vec{p},\vec{c}}$,
  where, for $j \in \nat_{\lt m}$, $\alpha_j$ is defined by
  $\alpha_j = \pair{\frac{1}{z}}{{-}\frac{a_j}{z}}$.
  
  We now repeatedly alter the tuples $\vec{\alpha}$, $\vec{p}$, $\vec{c}$
  until conditions~\ref{item:no:ambiguities} and~\ref{item:no:empty:cycles} are fulfilled
  while condition~\ref{item:form} is upheld.
  For this we let $n_0 \in \nat$, $w \in\two^*$, $m \in \nat$, $\vec{\alpha}$, $\vec{p}$, and $\vec{c}$ with 
  $\length{\vec{\alpha}} = \length{\vec{p}} = \length{\vec{c}} = m$
  be arbitrary such that \ref{item:form} holds. 
%

  First note that, if $m = 1$ and condition~\ref{item:no:ambiguities} or \ref{item:no:empty:cycles}
  are violated, then $\sigma \in \botdeg$, contradicting the assumption.

  In case \ref{item:no:ambiguities} does not hold, 
  consider the smallest $h \in \nat_{<m}$ 
  such that $c_h^\omega = p_{h+1} c_{h+1}^\omega$ 
  where addition in the subscripts is computed modulo~$m$.
  We assume $h < m-1$; the case $h = m-1$ proceeds analogously, 
  using Lemma~\ref{lem:doubleprod:rotation}.
  For $i \in \nat$ and $j \in \nat_{<m}$,
  we let $\cyc{i}{j} = (\wprod{\vec{\alpha}}{\shift{n_0}{f}})(mi + j)$.
  By Lemma~\ref{lem:merge} there are integers $a,b \ge 0$ and a word $x \in \two^*$
  such that $c_h^{\cyc{i}{h}} p_{h+1} c_{h+1}^{\cyc{i}{h+1}} = p_{h+1} x^{a \cyc{i}{h} + b \cyc{i}{h+1}}$~($\star$).
  We now define a tuple of weights $\vec{\beta}$, 
  and tuples of words $\vec{q}$ and $\vec{d}$, 
  with $\length{\vec{\beta}} = \length{\vec{q}} = \length{\vec{d}} = m - 1$,
  as follows: Let $j \in \nat_{<m-1}$.
  If $j < h$, then we define $q_j = p_j$, $d_j = c_j$, and $\beta_j = \alpha_j$.
  If $j > h$, we define $q_j = p_{j+1}$, $d_j = c_{j+1}$, and $\beta_j = \alpha_{j+1}$.
  If $j = h$, we define $q_j = p_h p_{h+1}$, $d_j = x$, and we let
  the weight $\beta_j$ be defined as follows:
  For $\alpha_h = \tup{r_0,r_1,\ldots,r_{k-1},e}$
  and $\alpha_{h+1} = \tup{r'_0,r'_1,\ldots,r'_{\ell-1},e'}$,
  let $\beta_h = \tup{a r_0, a r_1, \ldots, a r_{k-1}, b r'_0, b r'_1, \ldots, b r'_{\ell - 1}, a e + b e'}$.
%
  By definition of $\vec{q}$, $\vec{d}$, and $\vec{\beta}$,
  to verify $\dpf{\shift{n_0}{f},\vec{\alpha},\vec{p},\vec{c}} = \dpf{\shift{n_0}{f},\vec{\beta},\vec{q},\vec{d}}$,
  it suffices to check, for all $i \in \nat$,
  $p_h c_h^{\cyc{i}{h}} p_{h+1} c_{h+1}^{\cyc{i}{h+1}} = q_h d_h^{\cycp{i}{h}}$;
  here, for $i \in \nat$ and $j \in \{0,\ldots,m-2\}$,
  $\cycp{i}{j}$ is defined by $\cycp{i}{j} = (\wprod{\vec{\beta}}{\shift{n_0}{f}})((m-1)i+j)$.
  Fix $i \in \nat$.
  By definition of weighted products we have 
  $\cyc{i}{h} = \wof{\alpha_h}{\shift{t}{f}}$
  and $\cyc{i}{h+1} = \wof{\alpha_{h+1}}{\shift{t + k}{f}}$,
  for some $t \in \nat$ (see Lemma~\ref{lem:wprod:nth}).
  Also we have $\cycp{i}{h} = \beta_h \cdot \shift{t'}{f}$ for some $t' \in \nat$.
  By definition of $\vec{\beta}$ we obtain $t' = t$.
  It follows that 
  $\cycp{i}{h} = a \cdot \cyc{i}{h} + b \cdot \cyc{i}{h+1}$, 
  and we conclude by~($\star$). 
  Repeat the procedure with $\vec{\beta}$, $\vec{q}$, $\vec{d}$.

  In case \ref{item:no:empty:cycles} does not hold 
  because of $c_h = \emptyword$ for some $h \in \nat_{<m}$,
  we change $\alpha_h$ into a constant weight~$\tup{0,0,\ldots,0}$ 
  of length~$\length{\alpha_h}$. This clearly does not change~$\sigma$.
  Now consider the case that \ref{item:no:empty:cycles} does not hold 
  because $\alpha_h$ is constant.
  Then let $h \in \nat_{<m}$ be minimal with this property.
  We assume $h < m-1$. 
  The case $h = m-1$ proceeds analogously, again by Lemma~\ref{lem:doubleprod:rotation}.
  We now define a tuple of weights $\vec{\beta}$, 
  and tuples of words $\vec{q}$ and $\vec{d}$, 
  with $\length{\vec{\beta}} = \length{\vec{q}} = \length{\vec{d}} = m - 1$,
  as follows:
  For $j < h$, we define $\beta_j = \alpha_j$, $q_j = p_j$, and $d_j = c_j$.
  For $j > h$, we define $\beta_j = \alpha_{j+1}$, $q_j = p_{j+1}$, and $d_j = c_{j+1}$.
  For the case $j = h$, 
  let $\alpha_h = \tup{0,0,\ldots,0,e}$
  and $\alpha_{h+1} = \tup{r_0,r_1,\ldots,r_{\ell-1},e'}$.
  We define $q_j = p_h c_h^e p_{h+1}$, 
  $d_j = c_{h+1}$, and 
  $\beta_j = \tup{0,0,\ldots,0,r_0,r_1,\ldots,r_{\ell-1},e'}$ of length $\length{\alpha_h} + \length{\alpha_{h+1}} - 1$.
  The verification of $\dpf{\shift{n_0}{f},\vec{\alpha},\vec{p},\vec{c}} = \dpf{\shift{n_0}{f},\vec{\beta},\vec{q},\vec{d}}$
  is similar as above.
  Repeat the procedure with $\vec{\beta}$, $\vec{q}$, $\vec{d}$.
  \qed  
\end{proof}

For the proof of the following theorem we 
allow for a more liberal version of transducers.
Instead of input letters along the edges we now allow input \emph{words}.
Transitions of these transducers are of the form $q \stackrel{\pair{u}{v}|w}{\longrightarrow} q'$.
The idea is that this transition is taken if the automaton is in state~$q$ and
the input word is of the form $uv\tau$.
Then the automaton produces output $w$ and switches to state $q'$, consuming $u$ and continuing with $v\tau$.

\begin{definition}\label{def:LFST}
  An FST with \emph{look-ahead} (\lfst) is a tuple~$T = \tup{Q,q_0,D,\delta,\lambda}$
  where $Q$ is a finite set of states, $q_{0} \in Q$ is the initial state,
  the finite set $D \subseteq Q \times \two^{+} \times \two^{*}$ 
  is the input domain of the transition function
  $\delta : D \to Q$,
  and the output function $\lambda : D \to \two^*$, 
  satisfying the following condition:
  for all $q \in Q$, $u_1,u_2,v_1,v_2 \in \two^*$ if $u_1u_2$ is a prefix of $v_1v_2$ 
  and $\tup{q,u_1,u_2} \in D$ and $\tup{q,v_1,v_2} \in D$, then $u_1 = v_1$ and $u_2 = v_2$.

  We lift $\delta$ to a partial function $\delta^\star : Q \times \two^* \pto Q$ by
  $\delta^\star(q,\emptyword) = q$ and
  \begin{align*}
    \delta^\star(q,u_1u_2v) & = \delta^\star(\delta(q,u_1,u_2),u_2v)
    &&&& (\tup{q,u_1,u_2} \in D, v \in \two^*)\,.
  \end{align*}
  Similarly, we lift $\lambda$ to a partial function
  $\lambda^\star : Q \times \two^\infty \pto \two^\infty$ by
  $\lambda^\star(q,\emptyword) = \emptyword$ and
  \begin{align*}
    \lambda^\star(q,u_1u_2v) & = \lambda(q,u_1,u_2) \cdot \lambda^\star(\delta(q,u_1,u_2),u_2v)
    &&&& (\tup{q,u_1,u_2} \in D, v \in \two^\infty)\,.
  \end{align*}
  The partial function $T : \two^\infty \pto \two^\infty$ \emph{realised} by the \lfst~$T$
  is defined by $T(u) = \lambda^\star(q_0,u)$,
  for all $u \in \two^\infty$.
\end{definition}

These transducers can be simulated by FSTs.

\begin{lemma}\label{lem:LFST2FST}
  For every \lfst~$T$ there is an FST~$T'$ such that 
  for all $u \in \two^\infty$, 
  $T'(u) = T(u)$ whenever $T(u)$ is defined.
\end{lemma}
\begin{proof}
  Let $T = \tup{Q,q_0,D,\delta,\lambda}$ be an \lfst. 
  We define an FST $T' = \tup{Q',q'_0,\delta',\lambda'}$ as follows.
  Define $\ell = \max \{\length{u_1u_2} \where \myex{q}{\tup{q,u_1,u_2} \in D} \}$.
  We choose $Q' = Q \times \two^{\le \ell}$ and $q'_0 = \pair{q_0}{\emptyword}$ 
  where $\two^{\le \ell}$ are all words of length at most~$\ell$.
  For every $q\in Q$ and $a \in \two$ we define:
  \begin{align*}
    \pair{\delta'}{\lambda'}(\pair{q}{v},a) &= \pair{\pair{q}{va}}{\emptyword} && \text{$v \in \two^*$ with $\length{v} < \ell$}\\
    \pair{\delta'}{\lambda'}(\pair{q}{u_1u_2v},a) &= \pair{\pair{q}{u_2va}}{\lambda(q,u_1,u_2)} && \text{$u_1,u_2,v \in \two^*$, $\tup{q,u_1,u_2} \in D$}\\[-.5ex]
     &&& \text{with $\length{u_1u_2v} = \ell$}
  \end{align*}
  To make the FST $T'$ complete, we let $\pair{\delta'}{\lambda'}(\pair{q}{v},a) = \pair{\pair{q}{v}}{\emptyword}$ whenever $v \in \two^*$ and $\length{v} = \ell$
  and neither of the two clauses above applies.
  It is straightforward to verify that 
  for all $u \in \two^\infty$, 
  $T'(u) = T(u)$ whenever $T(u)$ is defined.
  \qed
\end{proof}

\begin{theorem}\label{thm:transducts}
  Let $f : \nat \to \nat$ be spiralling, and $\sigma \in \str{\two}$.
  Then $\seq{f} \fstred \sigma$
  if and only if 
  $\sigma \fstconv \seq{\wprod{\vec{\alpha}}{\shift{n_0}{f}}}$
  for some integer~$n_0 \ge 0$, and a tuple of weights $\vec{\alpha}$.
\end{theorem}

\begin{proof}
  One direction is by Lemma~\ref{lem:wprod:FST}.
  For the other, 
  assume $\seq{f} \fstred \sigma$.
  If $\sigma \in \botdeg$, then $\sigma \fstconv \seq{\wprod{\tup{\tup{0,0}}}{\shift{0}{f}}} = \seq{n \mapsto 0} = 1^\omega$.
  Thus let $\sigma \not\in \botdeg$.
  By Lemma~\ref{lem:disambiguate} there exist
  $n_1, m \in \nat$, $w \in \two^*$, $\vec{\alpha}$, $\vec{p}$ and $\vec{c}$ 
  with $\length{\vec{\alpha}} = \length{\vec{p}} = \length{\vec{c}} = m > 0$
  such that $\sigma = w \cdot \dpf{\shift{n_1}{f},\vec{\alpha},\vec{p},\vec{c}}$, and
  fulfilling the conditions~\ref{item:no:ambiguities} and~\ref{item:no:empty:cycles} 
  of Lemma~\ref{lem:disambiguate}.
  We abbreviate $g = \wprod{\vec{\alpha}}{\shift{n_1}{f}}$.
  We will show that $\sigma \fstconv \seq{g}$.
  By Lemma~\ref{lem:wprod:preserve:spiralling}
  we have that the function~$g$ is spiralling too.

  By conditions~\ref{item:no:ambiguities} and~\ref{item:no:empty:cycles},
  for every $j \in \nat_{<m}$,
  there exists $t_j \in \nat$ such that $c_j^\omega(t_j) \ne (p_{j+1}c_{j+1}^\omega)(t_j)$ 
  (where addition is modulo $m$);
  let $t_j$ be minimal with this property.
  For $j \in \nat_{<m}$, let $\ell_j,\ell_j' \in \nat$ be minimal such that $\length{c_j c_j^{\ell_j}} > t_j$
  and $\length{p_{j+1} c_{j+1}^{\ell_j'}} > t_j$.
  Then by minimality of $t_j$ and $\ell_j$, 
  we obtain
  \begin{enumerate}
    \item \label{item:thm:transducts:i}
      $c_j c_j^{\ell_j-1} \prefixof p_{j+1} c_{j+1}^{\ell'_j} \tau$ for every $\tau \in \str{\two}$, and
    \item 
      $c_j c_j^{\ell_j} \not\prefixof p_{j+1} c_{j+1}^{\ell'_j} \tau$ for every $\tau \in \str{\two}$
  \end{enumerate}
  (with again addition 
  computed modulo~$m$).
  From \ref{item:thm:transducts:i} we moreover obtain
  \begin{enumerate}[resume]
    \item \label{item:thm:transducts:iii}
      $c_j c_j^{\ell_j} \prefixof c_j^{n} p_{j+1} c_{j+1}^{\ell'_j} \tau$ for every $n > 0$ and $\tau \in \str{\two}$.
  \end{enumerate}

  Next, we take a suffix $\sigma'$ of $\sigma$ such that every occurrence of a block $p_{j+1} c_{j+1}^{\cyc{i}{j}}$
  has as a prefix $p_{j+1} c_{j+1}^{\ell_j}$.
  Let $n_2 \in \nat$ be such that for $g' = \shift{n_2 \cdot m}{g}$
  we have that $g'(n) > \max\{t_j \where j \in \nat_{<m}\}$ for all $n \in \nat$;
  the existence of such an $n_2$ follows from $g$ being spiralling.
  To prove
  $\sigma \fstconv \seq{g}$
  it suffices to show
  $\sigma' \fstconv \seq{g'}$
  where
  \begin{align*}
    \sigma' &= \prod_{i = n_2}^{\infty} \prod_{j=0}^{m-1} p_j c_j^{\cyc{i}{j}} 
     = \prod_{i = 0}^{\infty} \prod_{j=0}^{m-1} p_j c_j^{\cycp{i}{j}} 
    &
    \seq{g'} &= \prod_{i = 0}^{\infty} \prod_{j=0}^{m-1} 1 0^{\cycp{i}{j}} 
  \end{align*}
  where $\cyc{i}{j} = g(mi + j)$, and $\cycp{i}{j} = g'(mi + j)$.
  Note that by the choice of $n_2$, we have
  $\cycp{i}{j+1} \ge \ell'_j$ for all $i \in \nat$ and $j \in \nat_{<m}$.

  It is clear how to construct an FST that transduces $\seq{g'}$ to $\sigma'$,
  see Figure~\ref{fig:easy}.
  For $\sigma' \fstred \seq{g'}$, we define a \lfst~$T = \tup{Q,q_{m-1},D,\delta,\lambda}$, as follows,
  and apply Lemma~\ref{lem:LFST2FST}.
  Let $Q = \{q_j \where j \in \nat_{<m}\}$ and $D = \{ \tup{q_j,c_j,c_j^{\ell_j}} \where j \in \nat_{<m}\} \cup \{ \tup{q_j,p_{j+1},c_{j+1}^{\ell'_j}} \where j \in \nat_{<m}\}$,
  and define $\delta,\lambda$ by
  \begin{align*}
    \pair{\delta}{\lambda}(q_j,c_j,c_j^{\ell_j}) = \pair{q_j}{0} &&&
    \pair{\delta}{\lambda}(q_j,p_{j+1},c_{j+1}^{\ell_j'}) = \pair{q_{j+1}}{1} \,.
  \end{align*}
  We now argue that $\sigma' \fstred_T \seq{g'}$.
  This follows from the following facts:
  \begin{enumerate}[label=(\alph*)]
    \item $\lambda^\star(q_j,p_{j+1}c^{\cycp{i}{j+1}} \tau) = 1 \cdot \lambda^\star(q_{j+1}, c^{\cycp{i}{j+1}} \tau)$,
    \item $\lambda^\star(q_j,c_j^{n}p_{j+1}c^{\cycp{i}{j+1}} \tau) = 0 \cdot \lambda^\star(q_{j+1}, c^{\cycp{i}{j+1}} \tau)$
      for all $n > 0$, since by item~\ref{item:thm:transducts:iii} 
      we have $c_j c_j^{\ell_j} \prefixof c_j^{n}p_{j+1}c^{\cycp{i}{j+1}} \tau$.
      \qed
  \end{enumerate}
\end{proof}

\begin{lemma}\label{lem:one:weight}
  Let $f : \nat \to \nat$ be spiralling, and $\sigma \not\in \botdeg$ with $\seq{f} \fstred \sigma$.
  Then we have $\sigma \fstred \seq{\wprod{\tup{\beta}}{\shift{n_0}{f}}}$
  for some integer~$n_0 \ge 0$, and a non-constant weight $\beta$.
\end{lemma}
\begin{proof}
  Let $\sigma \not\in \botdeg$ be a transduct of $\seq{f}$. 
  By Theorem~\ref{thm:transducts} we have
  $\sigma \fstconv \seq{\wprod{\vec{\alpha}}{\shift{n_0}{f}}}$
  for some~$n_0 \in \nat$, and an $m$-tuple of weights $\vec{\alpha}$.
  By Lemma~\ref{lem:wprod:not:botdeg}, there exists $i \in \nat_{<m}$ 
  such that $\alpha_i = \tup{a_0,a_1,\ldots,a_{k_i-1},b_i}$ is not constant.
  Let $k_j = \length{\alpha_j}-1$ ($j \in \nat_{<m}$), 
  and for $\ell \in \nat_{\le m}$ let $s_\ell = \sum_{j=0}^{\ell - 1} k_j$.
  We define a weight $\beta$ of length $s_m + 1$ where, for $j \in \nat_{<m}$ and $h \in \nat_{<k_j}$,
  $\beta_{s_j + h} = 0$ if $j \ne i$, and $\beta_{s_j + h} = a_h$ if $j = i$, 
  and $\beta_{s_m} = b_i$.
  Then we have
  $\sigma \fstred \seq{n \mapsto (\wprod{\vec{\alpha}}{\shift{n_0}{f}})(nm + i)} = \seq{\wprod{\tup{\beta}}{\shift{n_0}{f}}}$.
\end{proof}

\begin{theorem}\label{thm:no:minimal:below}
  There is a non-atom, non-zero degree $\convclass{\sigma}$ that has no atom degree below it.
  Hence, non-zero transducts of $\sigma$ start an infinite~descending~chain. 
\end{theorem}

\begin{proof}
  We define the function $f : \nat \to \nat$ by $f(n) = 2^n$.
  We show that the degree $\convclass{\seq{f}}$ has no atom degree below it.
  Let $\sigma\not\in\botdeg$ with $\seq{f} \fstred \sigma$. 
  By Lemma~\ref{lem:one:weight} there is a non-constant weight $\beta = \tup{a_0,a_1,\ldots,a_{k-1},b}$
  such that $\sigma \fstred \seq{g}$ where 
  $g = \wprod{\tup{\beta}}{\shift{n_0}{f}}$.
  Since $f(n) = 2^n$ it follows that 
  \begin{align*}
    g(n) = b+\sum_{i =0}^{k-1} a_i 2^{n_0+nk+i} = b + 2^{nk}\sum_{i =0}^{k-1} a_i 2^{n_0 +i} =  (g(0)-b) \cdot 2^{nk} + b \,.
  \end{align*}
  By Lemma~\ref{lem:wprod:FST} we have that $\seq{g} = \seq{(g(0)-b) \cdot 2^{nk} + b} \fstconv \seq{2^{nk}}$.
  Thus we have $\sigma \fstred \seq{2^{nk}}$.
  Also $\seq{2^{nk}} \fstred \seq{2^{2nk}}$
  holds by Lemma~\ref{lem:basic}~\ref{item:lem:basic:sub},
  and by Lemma~\ref{lem:too:fast} we conclude $\seq{2^{2nk}} \not\fstred \seq{2^{nk}}$.
  \qed
\end{proof}

\section{Squares}\label{sec:squares}
In~\cite{endr:hend:klop:2011} it is shown that $\convclass{\seq{n}}$ is an atom degree.
One of the main questions of~\cite{endr:hend:klop:2011} is whether there exist other atom degrees. 
Here we show that also $\convclass{\seq{n^2}}$ is an atom degree.
The main tool is Theorem~\ref{thm:transducts}, 
the characterisation of transducts of spiralling sequences,
which implies the following proposition.

\begin{proposition}\label{prop:poly-transducts}
  Let $p(n)$ be a polynomial of degree $k$ with non-negative integer coefficients,
  and let $\sigma$ be a transduct of $\seq{p(n)}$ with $\sigma\notin\botdeg$.
  Then $\sigma \fstred \seq{q(n)}$
  for some polynomial $q(n)$ of degree~$k$
  with non-negative integer coefficients.
\end{proposition}

\begin{proof}
  By Lemma~\ref{lem:one:weight} it follows that
  $\sigma \fstred \seq{\wprod{\tup{\alpha}}{\shift{n_0}{p}}}$
  for some integer~$n_0 \ge 0$, and a non-constant weight $\alpha = \tup{a_{0},\ldots,a_{k-1},b}$.
  Let $h \in \nat_{<k}$ be such that  $a_{h} \neq 0$.
  Then we find
  $
    (\wprod{\tup{\alpha}}{\shift{n_0}{p}})(n)
      =
    b + \sum_{j=0}^{k-1}  a_{j} \cdot p(n_0 + nk + j)
  $, 
  which 
  can easily 
  be recognised to be a polynomial $q(n)$ of degree~$k$.
  \qed
\end{proof}

\begin{theorem}\label{thm:squares}
  The degree $\convclass{\seq{n^2}}$ is an atom.
\end{theorem}
\begin{proof}
  Let $\sigma \in \str{\two}$ be a transduct of $\seq{n^2}$
  such that $\sigma$ is not ultimately periodic.
  By Proposition~\ref{prop:poly-transducts} 
  there are integers $a > 0$, $b,c \ge 0$ such that $\sigma \fstred \seq{an^2 + bn + c}$.
  We first assume $2a \ge b$.
  Abbreviate $f(n) = a n^2 + (2a+b)n$.
  The roman numbers below refer to Lemma~\ref{lem:basic}.
  We derive
  \begin{align*}
    \seq{a n^2 + bn + c}  
    & \fstconv \seq{a n^2 + bn}
    && \text{by \ref{item:lem:basic:yshift}} \\
    & \fstconv \seq{a (n+1)^2 + b(n+1)}
    && \text{by \ref{item:lem:basic:xshift}} \\
    & \fstconv \seq{f(n)}
    && \text{by \ref{item:lem:basic:yshift}} \\
    & \fstred \seq{b(f(2n)) + (2a-b)(f(2n+1))}
    && \text{by \ref{item:lem:basic:merge}} \\
    & \fstconv \seq{8 a^2 n^2 + 16 a^2 n} \fstconv \seq{8 a^2 (n+1)^2}
    && \text{by \ref{item:lem:basic:yshift}} \\
    & \fstconv \seq{(n+1)^2} 
    && \text{by \ref{item:lem:basic:mult}} \\
    & \fstconv \seq{n^2} 
    && \text{by \ref{item:lem:basic:xshift}} \,.
  \end{align*}
  If $2a < b$, choose $d$ such that $2ad \ge b$. 
  Then we have $\seq{an^2 + bn} \fstred \seq{ad^2n^2 + bdn}$ by~\ref{item:lem:basic:sub},
  and we reason as above for $\seq{a'n^2 + b'n}$ with $a' = ad^2$ and $b' = bd$.

  This shows that every non-ultimately periodic transduct of $\seq{n^2}$
  can be transduced back to $\seq{n^2}$.
  Hence, the degree of $\seq{n^2}$ is an atom.
  \qed
\end{proof}
%
%

\bibliography{main}

\end{document}